%% file: document_full.tex
\documentclass[a4paper]{scrartcl}
\usepackage{amsthm}
\usepackage{amsmath}
\usepackage{amssymb}
\usepackage{cite}
\usepackage{hyperref}
\usepackage{color}
\usepackage{xspace}
\usepackage{authblk}

\newcommand{\argmin}{\mathrm{argmin}}
\newcommand{\dunion}{\ \dot\cup\ }

\newtheorem{theorem}{Theorem}
\newtheorem{lemma}[theorem]{Lemma}

\newtheorem{corollary}[theorem]{Corollary}
\newtheorem{definition}[theorem]{Definition}

\input{macros}

\begin{document}

\title{A Definability Dichotomy for Finite Valued CSPs}
\author{Anuj Dawar and Pengming Wang}
\affil{University of Cambridge Computer Laboratory\\
	\texttt{\{anuj.dawar, pengming.wang\}@cl.cam.ac.uk}}

\maketitle

\begin{abstract}
Finite valued constraint satisfaction problems are a formalism for
describing many natural optimization problems, where constraints on
the values that variables can take come with rational weights and the
aim is to find an assignment of minimal cost.  Thapper and \v{Z}ivn\'y
have recently established a complexity dichotomy for finite valued constraint
languages.  They show that each such language either gives rise to a
polynomial-time solvable optimization problem, or to an $\NP$-hard
one, and establish a criterion to distinguish the two cases.  We
refine the dichotomy by showing that all optimization problems in the
first class are definable in fixed-point language with counting, while
all languages in the second class are not definable, even in
infinitary logic with counting.  Our definability dichotomy is not
conditional on any complexity-theoretic assumption.
\end{abstract}
\section{Introduction} \label{sec:intro}
\input{intro}

\section{Background}\label{sec:background}

\input{background}

\section{Definable Reductions} \label{sec:reductions}

\input{reductions_full}

\section{Expressibility Result}\label{s:exp}

\input{expressibility_full}

\section{Inexpressibility Result}\label{s:inexp}
\input{inexpressibility_full}

\section{Constraint Languages of Infinite Size}\label{s:inf}

\input{inf_size_full}

%\bibliography{ref}

\input{document_full.bbl}
\end{document}

%% file: macros.tex
\newcommand{\QQ}{\ensuremath{\mathbb{Q}}}
\newcommand{\NN}{\ensuremath{\mathbb{N}}}
\newcommand{\ZZ}{\ensuremath{\mathbb{Z}}}

\newcommand{\PT}{\ensuremath{\mathrm{P}}\xspace}
\newcommand{\NP}{\ensuremath{\mathrm{NP}}\xspace}
\newcommand{\XOR}{\ensuremath{\mathrm{XOR}}\xspace}
\newcommand{\infC}{\ensuremath{C^{\omega}}\xspace}
\newcommand{\VCSP}{\ensuremath{\mathrm{VCSP}\xspace}}
\newcommand{\CSP}{\ensuremath{\mathrm{CSP}\xspace}}
\newcommand{\BLP}{\ensuremath{\mathrm{BLP}\xspace}}
\newcommand{\LP}{\ensuremath{\mathrm{LP}\xspace}}
\newcommand{\BIT}{\ensuremath{\mathrm{BIT}\xspace}}
\newcommand{\MULT}{\ensuremath{\mathrm{MULT}\xspace}}

\newcommand{\MAXCUT}{\ensuremath{\mathrm{MAXCUT}\xspace}}
\newcommand{\SAT}{\ensuremath{\mathrm{SAT}\xspace}}
\newcommand{\NAE}{\ensuremath{\mathrm{NAE}\xspace}}
\newcommand{\NAESAT}{\ensuremath{\mathrm{NAESAT}\xspace}}

\newcommand{\logic}[1]{\text{\upshape #1}\xspace}
\newcommand{\FPC}{\logic{FPC}}
\newcommand{\FO}{\logic{FO}}
\newcommand{\Dl}{\logic{Datalog}}

\newcommand{\ar}{\ensuremath{\mathrm{ar}}}
\newcommand{\expGamma}{\langle\Gamma\rangle}

\newcommand{\ra}{\rightarrow}

\newcommand{\logicoperator}[1]{\mathbf{#1}}
\newcommand{\ifpop}{\logicoperator{ifp}}

\renewcommand{\vec}[1]{\bar{#1}}
\newcommand{\tup}[1]{\vec{#1}}
\newcommand{\countingTerm}[1]{\#_{#1}}

\newcommand{\struct}[1]{\ensuremath{\mathbf #1}\xspace}
\newcommand{\univ}[1]{\ensuremath{\dom(\struct #1)}} 

\newcommand{\dom}{\mathrm{dom}}
\newcommand{\vocvec}{\tau_{\text{vec}}}
\newcommand{\vocmat}{\tau_{\text{mat}}}

\newcommand{\defeq}{:=}
\newcommand{\var}{\mathit{var}}
\newcommand{\con}{\mathit{con}}
\newcommand{\dm}{\mathit{dom}}

\newcommand{\Def}{\mathit{Def}}
\newcommand{\Enc}{\mathit{Enc}}

%% file: intro.tex
Constraint Satisfaction Problems (CSPs) are a widely-used formalism
for describing many problems in optimization, artificial intelligence
and many other areas.  The classification of CSPs according to their
tractability has been a major area of theoretical research ever since
Feder and Vardi~\cite{FV98} formulated their dichotomy conjecture.
The main aim is to classify various constraint satisfaction problems
as either tractable (i.e.\ decidable in polynomial time) or $\NP$-hard
and a number of dichotomies have been established for special cases of
the CSP as well as generalizations of it.  In particular, Cohen et al.~\cite{ccjk06:ai} extend the algebraic methods that have been very successful in the classification of CSPs to what they call \emph{soft constraints}, that is constraint problems involving optimization rather than decision problems.  In this context, a  recent 
result by Thapper and  \v{Z}ivn\'y~\cite{tz13:stoc} established a
complexity dichotomy for \emph{finite valued CSPs} (VCSPs).  This is a
formalism for defining optimization problems that can be expressed as
sums of explicitly given rational-valued functions (a more formal
definition is given in Section~\ref{sec:background}).  As Thapper
and  \v{Z}ivn\'y argue, the formalism is general enough to include a
wide variety of natural optimization problems.  They show that every
finite valued CSP is either in $\PT$ or $\NP$-hard and provide a
criterion, in terms of the existence of a definable $\XOR$ function,
that determines which of the two cases holds.

In this paper we are interested in the \emph{definability} of
constraint satisfaction problems in a suitable logic.  Definability in
logic has been a significant tool for the study of CSPs for many
years. A particular logic that has received attention in this
context is $\Dl$, the language of inductive definitions by
function-free Horn clauses.  A dichotomy of definability has been
established in the literature, which shows that every constraint
satisfaction problem on a fixed template is either definable in $\Dl$
or it is not definable even in the much stronger $\infC$---an
infinitary logic with counting.  This result has not been published as
such but is an immediate consequence of results in~\cite{abd09:tcs}
where it is shown that every CSP satisfying a certain algebraic
condition is not definable in $\infC$, and in~\cite{BartoK14} where it
is shown that those that fail to satisfy this condition have bounded
width and are therefore definable in $\Dl$.  The definability
dichotomy so established does not line up with the (conjectured)
complexity dichotomy as it is known that there are tractable CSPs that
are not definable in $\Dl$.

In the context of the definability of optimization problems, one needs to distinguish three kinds of definability.  In general an optimization problem asks for a \emph{solution} (which will typically be an assignment of values from some domain $D$ to the variables $V$ of the instance) minimising the value of a \emph{cost} function.  This problem is standardly turned into a decision problem by including a budget $b$ in the instance and asking if there is a solution that achieves a cost of at most $b$.  Sentences in a logic naturally define decision problems, and in the context of definability a natural question is whether the decision problem is definable.  Asking for a formula that defines an actual optimal solution may not be reasonable as such a solution may not be uniquely determined by the instance and formulas in logic are generally invariant under automorphisms of the structure on which they are interpreted.  
An intermediate approach is to ask for a term in the logic that defines the \emph{cost} of an optimal solution and this is our approach in this paper.

Our main result is a definability dichotomy for finite valued
CSPs.  In the context of optimization problems involving numerical
values, $\Dl$ is unsuitable so we adopt as our yardstick definability
in fixed-point logic with counting ($\FPC$).  This is an important logic
that defines a natural and powerful proper fragment of the
polynomial-time decidable properties (see~\cite{Daw15}).  It should be
noted that $\infC$ properly extends the expressive power of $\FPC$ and
therefore undefinability results for the former yield undefinability
results for the latter.  We establish that every finite valued CSP is
either definable in $\FPC$ or undefinable in $\infC$.  Moreover, this
dichotomy lines up exactly with the complexity dichotomy of Thapper
and \v{Z}ivn\'y.  All the valued CSPs they determine are tractable are
in fact definable in $\FPC$, and all the ones that are $\NP$-hard are
provably not in  $\infC$.  It should be emphasised that, unlike the
complexity dichotomy, our definability dichotomy is not conditional on
any complexity-theoretic assumption.  Even if it were the case that
$\PT = \NP$, the finite valued CSPs still divide into those definable
in $\FPC$ and those that are not on these same lines.

The positive direction of our result builds on the recent work of
Anderson et al.~\cite{adh13:lics} showing that solutions to explicitly
given instances of linear programming are definable in $\FPC$. 
Thapper and  \v{Z}ivn\'y show that for the tractable VCSPs the optimal
solution can be found by solving their basic linear programming (BLP)
relaxation.  Thus, to establish the definability of these problems in
$\FPC$ it suffices to show that the reduction to the BLP is itself
definable in $\FPC$, which we do in Section~\ref{s:exp}.

For the negative direction, we use the reductions used
in~\cite{tz13:stoc} to establish $\NP$-hardness of VCSPs and show that
these reductions can be carried out within $\FPC$.  We start with the
standard CSP form of 3-SAT, which is not definable in $\infC$ as a
consequence of results from~\cite{abd09:tcs}.  Details of all these
reductions are presented in Section~\ref{s:inexp}.

There is one issue with regard to the representation of instances of VCSPs as relational structures which we need to consider in the context of definability.  An instance is defined over a language which consists of a set $\Gamma$ of functions from a finite domain $D$ to the rationals.  If $\Gamma$ is a finite set, it is reasonable to fix the relational signature to have a relation for each function in $\Gamma$, and the $\FPC$ formula defining the class of VCSPs would be in this fixed relational signature.   Indeed, the result of Thapper and \v{Z}ivn\'y~\cite{tz13:stoc} is stated for infinite sets $\Gamma$ but is really about finite subsets of it.  That is, they show that if $\Gamma$ does not have the $\XOR$ property, then \emph{every} finite subset of $\Gamma$ determines a tractable VCSP and that if $\Gamma$ does have the $\XOR$ property then it contains a finite subset $\Gamma'$ such that $\VCSP(\Gamma')$ is $\NP$-hard.  Our definability dichotomy replicates this precisely.
However, we can also consider the \emph{uniform} definability of $\VCSP(\Gamma)$ when $\Gamma$ is infinite (note that only finitely many functions from the language $\Gamma$ are used in constraints in any instance).  A natural way to represent this is to allow the functions themselves to be elements of the relational structure coding an instance.  We can show that our dichotomy holds even under this uniform representation.  For simplicity of exposition, we present the results for finite $\Gamma$ and then, in Section~\ref{s:inf}, we explain how the proof can be modified to the uniform case where the functions are explicitly given as elements of the structure. 
%\pengming{Should we add here our main theorem?}\marginnote{I don't
%think it necessary}

%% file: background.tex
	\textbf{Notation.}  We write $\NN$ for the natural numbers, $\ZZ$ for the integers, $\QQ$ for the rational numbers and $\QQ^+$ to denote the positive rationals.

We use bars $\bar{v}$ to denote vectors.  A vector over a set $A$ indexed by a set $I$ is a function $\bar{v}:I\rightarrow A$.  We write  $v_a$ for $\bar{v}(a)$.  Often, but not always, the index set $I$ is $\{1,\dots,d\}$, an initial segment of the natural numbers.  In this case, we also write $|\vec{v}|$ for the length of $\vec{v}$, i.e.\ $d$.
A matrix $M$ over $A$ indexed by two sets $I,J$ is a function 
	$M:I\times J\rightarrow A$.  We use the symbol $\dunion$ for the disjoint
	union operator on sets.

If $\bar{v}$ is an $I$-indexed vector over $A$ and $f: A \ra B$ is a function, we write $f(\bar{v})$ to denote the $I$-indexed vector over $B$ obtained by applying $f$ componentwise to $\bar{v}$.
	
	\subsection{Valued Constraint Satisfaction}
We begin with the basic definitions of valued constraint satisfaction problems.  These definitions are based, with minor modifications, on the definitions given in~\cite{tz13:stoc}.

	\begin{definition}%[Valued Constraint Language]
		Let $D$ be a finite domain. A \emph{valued constraint language} $\Gamma$
		over $D$ is a set of functions, where each $f \in \Gamma$ has an associated arity $m = \ar(f)$  and $f:D^{m} \rightarrow \mathbb{Q}^+$.
	\end{definition}
	
	\begin{definition}%[Valued Constraint Satisfaction Problem]
		An instance of the \emph{valued constraint satisfaction problem (\VCSP)}
		over a valued constraint language $\Gamma$ is a pair $I=(V,C)$,
		where $V$ is a finite set of \emph{variables} and $C$ is a finite set
		of \emph{constraints}. Each constraint
		in $C$ is a triple $(\sigma, f, q)$, where 
		$f\in \Gamma$, $\sigma \in V^{\ar(f)}$ and $q \in \QQ$.
		
		A \emph{solution} to an instance $I$ of $\VCSP(\Gamma)$ is an assignment 
		$h:V\rightarrow D$ of values in $D$ to the variables in $V$.  The \emph{cost} of the solution $h$ is given by
		$cost_I(h):=\sum_{(\sigma, f, q)\in C} q\cdot f(h(\sigma))$. 
 The valued constraint satisfaction problem is then to 
		find a solution with minimal cost.
		
		In the \emph{decision version} of the problem, an additional threshold
		constant $t\in\mathbb{Q}$ is given, and the question becomes whether there
		is a solution $h$ with $cost_I(h)\leq t$.
	\end{definition}
	
Given a valued constraint language $\Gamma$, there are certain natural closures $\Gamma'$ of this set of functions for which the computational complexity of $\VCSP(\Gamma)$ and $\VCSP(\Gamma')$ coincide.  The first we consider is called the \emph{expressive power} of $\Gamma$, which consists of functions that can be defined by minimising a cost function over a fixed $\VCSP(\Gamma)$-instance $I$ over some projection of the variables in $I$ (this is defined formally below).  The second closure of $\Gamma$ we consider is under scaling and translation.  Both of these are given formally in the following definition.
	\begin{definition}%[Expressive Power, Scaling, and Translation]
\label{def:expower}
		Let $\Gamma$ be a valued constraint language over $D$. We say a function
		$f:D^m\rightarrow\mathbb{Q}$, is 
		\emph{expressible} in $\Gamma$, if there is some instance 
		$I_f=(V_f,C_f)\in \VCSP(\Gamma)$ and a tuple $\tup{v}=(v_1,\ldots,v_m)\in V^m_f$ such that
		\[f(\bar{x}) = \min_{h\in H_{\bar{x}}}
			 cost_{I_f}(h),\]
		where
		$H_{\bar{x}}:=\{h:V_f\rightarrow D \mid h(v_i)=x_i \ ,\ 1\leq i\leq m\}$.
		We then say the function $f$ is \emph{expressed} by the instance $I_f$ and
		the tuple $\tup v$, and 
		call the set of all functions that can be expressed by an instance of 
		$\VCSP(\Gamma)$ the
		\emph{expressive power} of $\Gamma$, denoted by $\langle\Gamma\rangle$.
		
		Furthermore, we write $f'\equiv f$ if $f'$ is obtained from $f$ by
		\emph{scaling} and \emph{translation}, i.e.\ there are $a,b\in\QQ, a>0$
		such that $f'=a\cdot f+b$. For a valued constraint language $\Gamma$, we write
		$\Gamma_{\equiv}$ to denote the set $\{f' \mid f'\equiv f\ \mbox{for some}\ f\in\Gamma\}$.
	\end{definition}

The next two lemmas establish that closing $\Gamma$ under these
operations does not change the complexity of the corresponding
problem.  The first of these is implicit in the literature, and we
prove a stronger version of it in Lemma~\ref{lem:equiv_to_gamma}.
	
	\begin{lemma}\label{lem:equiv}
		Let $\Gamma$ and $\Gamma'$ be valued constraint languages on domain $D$ such that
		$\Gamma' \subseteq \Gamma_\equiv$.  
		Then $\VCSP(\Gamma')$ is polynomial-time reducible to $\VCSP(\Gamma)$.
	\end{lemma}
	%% \begin{proof}
	%% 	Let $I=(V,C)$ be an instance of $\VCSP(\Gamma')$. Now, for each function
	%% 	$f\in\Gamma'$ there exists a $g\in\Gamma$ with some 
	%% 	$a_{f},b_{f}\in\QQ,a_{f}>0$ with $f=a_{f}\cdot g+b_{f}$. By
	%% 	replacing each constraint $(\sigma,f,q)\in C$ by the constraint 
	%% 	$(\sigma,g,a_f\cdot q)$, we obtain an instance of $\VCSP(\Gamma)$ that
	%% 	has the same optimal solutions as $I$.
	%% \end{proof}
	
	\begin{lemma}[Theorem 3.4, \cite{ccjk06:ai}]\label{lem:expower}
		Let $\Gamma$ and $\Gamma'$ be valued constraint languages on domain $D$ such that
		$\Gamma' \subseteq \langle\Gamma\rangle$.  
		Then $\VCSP(\Gamma')$ is polynomial-time reducible to $\VCSP(\Gamma)$.
	\end{lemma}

In the study of constraint satisfaction problems, and of structure
homomorphisms more generally the \emph{core} of a structure plays an
important role.  The corresponding notion for valued constraint
languages is given in the following definition.
	\begin{definition}%[Core]
		We call a valued constraint language $\Gamma$ over domain $D$
		a \emph{core} if for for all $a\in D$, there is some instance 
		$I_a\in \VCSP(\Gamma)$ such that in every minimal cost solution over $I_a$, some variable
		is assigned $a$.
		A valued constraint language $\Gamma'$ over a domain $D'\subseteq D$ 
		is a \emph{sub-language} of $\Gamma$ if it contains exactly the functions
		of $\Gamma$ restricted to $D'$. We say $\Gamma'$ is a \emph{core of} $\Gamma$,
		if $\Gamma'$ is a sub-language of $\Gamma$ and also a core.
	\end{definition}
	
	\begin{lemma}[Lemma 2.4, \cite{tz13:stoc}]\label{lem:core_equiv}
		Let $\Gamma'$ be a core of $\Gamma$. Then, 
		$\min_h cost_I(h)=\min_h cost_{I'}(h)$ for all $I\in \VCSP(\Gamma)$ and 
		$I'\in \VCSP(\Gamma')$ where $I'$ is obtained from $I$
		by replacing each function of $\Gamma$ by its restriction in $\Gamma'$. 
	\end{lemma}
	
Finally, we consider the closure of $\Gamma$ under parameterized
definitions.  That is, we define $\Gamma_c$, the language obtained
from $\Gamma$ by allowing functions that are obtained from those in
$\Gamma$ by fixing some parameters.
	\begin{definition}\label{def:gamma_c}
		Let $\Gamma$ be a core over $D$, we denote by $\Gamma_c$
		the language that contains exactly those functions 
		$f: D^m \ra \QQ$ for which there exists
		\begin{itemize}
		  \item a function $g\in\Gamma$, with $g: D^n \ra \QQ$
                  with $n\geq m$,
		  \item an injective mapping $s_{f}:\{1,\ldots,m\}\rightarrow\{1,\ldots,n\}$,
		  \item an index set $T_{f}\subseteq\{1,\ldots,n\}$,
		  \item and a partial assignment $t_{f}:T_{f}\rightarrow D$,
		\end{itemize}
		such that $f$ is $g$ restricted on $t_{f}$, i.e.\ 
		$f(x_{s_{f}(1)},\ldots,x_{s_{f}(m)})=f(t(x_1),\ldots,t(x_n))$,
		where $t(x_i)=t_{f}(i)$ if $i\in T_{f}$, and $t(x_i)=x_i$ otherwise.
		Furthermore, we fix a mapping $\gamma:\Gamma_c\rightarrow \Gamma$ that
		assigns each $f\in\Gamma_c$ a function $g=\gamma(f)\in\Gamma$ with the above properties.
	\end{definition}
		%% The language $\Gamma_c$ can be seen as the language obtained from $\Gamma$
		%% by adding all functions from $\Gamma$ where some subset of variables is
		%% fixed to some elements of the domain.
 For example, if $f(x_1,x_2,x_3)\in\Gamma$, 
		then $g(x_1,x_2):=f(x_1,a,x_2)$ for $a\in D$ is in $\Gamma_c$.

	\subsection{Linear Programming}
	
	\begin{definition}\label{def:lin-opt}%[Linear Optimization]
		Let $\mathbb{Q}^V$ be the rational Euclidean space indexed by a set $V$.
		A \emph{linear optimization problem} is given by a \emph{constraint matrix}
		$A\in\mathbb{Q}^{C\times V}$ and vectors $\bar{b}\in\mathbb{Q}^C, \bar{c}\in\mathbb{Q}^V$.
		Let $P_{A,\bar{b}}:=\{\bar{x}\in\mathbb{Q}^V | A\bar{x}\leq \bar{b}\}$ be the set of 
		\emph{feasible solutions}. The linear optimization problem is then to 
		determine either that $P_{A,\bar{b}}=\emptyset$, or to find a vector 
		$\bar{y}=\mathrm{argmax}_{\bar{x}\in P_{A,\bar{b}}} \bar{c}^T\bar{x}$, or to determine that
		$\max_{\bar{x}\in P_{A,\bar{b}}} \bar{c}^T\bar{x}$ is unbounded.
		
		We speak of the \emph{integer linear optimization problem}, if the set
		of feasible solutions is instead defined as 
		$P_{A,\bar{b}}:=\{\bar{x}\in\mathbb{Z}^V | A\bar{x}\leq \bar{b}\}$.
		
		In the decision version of the problem, an additional constant $t\in\mathbb{Q}$
		is given, and the task is determine whether there exists a feasible solution
		$\bar{x}\in P_{A,\bar{b}}$, such that $\bar{c}^T\bar{x}\geq t$.

	\end{definition}
	It is often convenient to describe the linear optimization problem  $(A,\bar{b},\bar{c})$ as a system of linear inequalities $A\bar{x}\leq\bar{b}$
	along with the objective $\max_{\bar{x}\in P_{A,\bar{b}}} \bar{c}^T\bar{x}$.  We may also alternatively, describe an instance with a minimization objective.  It is easy to see that such a system can be converted to the standard form of Defintion~\ref{def:lin-opt}.
	
	Let $\Gamma$ now be a valued constraint language over $D$, and let $I=(V,C)$ be
	an instance of $\VCSP(\Gamma)$. We associate with $I$ the following linear
	optimization problem in variables $\lambda_{c,\nu}$ for each $c \in C$ with $c=(\sigma,f,q)$ and $\nu \in D^{\ar(f)}$, and $\mu_{x,a}$ for each $x\in V$ and $a \in D$.
	\begin{equation}
		\min \sum_{c\in C}  \sum_{\nu\in D^{\ar(f)}} \!\! \lambda_{c,\nu} \cdot q \cdot f(\nu) \quad \text{where } c = (\sigma,f,q) 
  \end{equation}
%\marginnote{I have tried to re-write this system of equations in order to try and understand it, as I found the previous version very difficult.  Could you please check that my understanding is correct?}
subject to the following constraints.\\
For each $c\in C$ with $c = (\sigma,f,q)$, each $i$ with $1\leq i \leq \ar(f)$ and each $a\in D$, we have
\begin{equation}\label{eqn:constr}
		\sum_{\nu\in D^{\ar(f)}: \nu_i = a} \!\! \!\! \lambda_{c,\nu} \;
		= \; \mu_{\sigma_i,a} ; 
\end{equation}
for each $x \in V$, we have
\begin{equation}
\sum_{a\in D} \mu_{x,a} \; = \; 1;
\end{equation}
 and for all variables $\lambda_{c,\nu}$ and $\mu_{x,a}$ we have
\begin{equation}
0\leq \lambda_{c,\nu}\leq 1 \quad \text{and} \quad 0 \leq \mu_{x,a} \leq 1 .
\end{equation}

A feasible \emph{integer} solution to the above system defines a solution $h: V \ra D$ to the instance $I$, given by $h(x) = a$ iff $\mu_{x,a} = 1$.  Equations~\ref{eqn:constr} then ensure that $\lambda_{c,\nu} = 1$ for $c=(\sigma,f,q)$ just in case $h(\sigma) = \nu$.  Thus, it is clear that an optimal integer solution gives us an optimal solution to $I$.

 If we consider \emph{rational} solutions instead of integer ones, we  obtain the \emph{basic LP-relaxation} of $I$, which we denote $\BLP(I)$.  The following theorem characterises for which languages $\Gamma$
	$\BLP(I)$ has the same optimal solutions as $I$.
	
	For the statement of the dichotomy result from \cite{tz13:stoc}, we need to 
	introduce an additional notion.	
	We say the property $(\XOR)$ holds for a valued constraint language $\Gamma$
	over domain $D$ if there are $a,b\in D, a\neq b$, such that $\langle\Gamma\rangle$
	contains a binary function $f$ with $\argmin\ f=\{(a,b),(b,a)\}$.
	
	\begin{theorem}[Theorem 3.3, \cite{tz13:stoc}]\label{thm:dichotomy}
		Let $\Gamma$ be a core over some finite domain
		$D$.
		\begin{itemize}
		  \item Either for each instance $I$ of $\VCSP(\Gamma)$, the optimal solutions of $I$ are the same as $\BLP(I)$;
		  \item or property $(\XOR)$ holds for $\Gamma_c$ and $\VCSP(\Gamma)$ is
		  		$\NP$-hard.
		\end{itemize}
	\end{theorem}

	\subsection{Logic}

A relational \emph{vocabulary} (also called a \emph{signature} or a \emph{language}) $\tau$ is a finite sequence of relation and constant symbols $(R_1, \dots, R_k, c_1, \dots, c_l)$, where every relation symbol $R_i$ has a fixed \emph{arity} $a_i \in \NN$. A structure $\struct A = (\univ A, R_1^{\struct A}, \dots, R_k^{\struct A}, c_1^{\struct A}, \dots, c_l^{\struct A})$ over the signature $\tau$ (or a \emph{$\tau$-structure}) consists of a non-empty set $\univ A$, called the \emph{universe} of $\struct A$, together with relations $R_i^{\struct A} \subseteq \univ A^{a_i}$ and constants $c_j^{\struct A} \in \univ A$ for each $1 \leq i \leq k$ and $1 \leq j \leq l$. Members of the set $\univ A$ are called the \emph{elements} of $\struct A$ and we define the \emph{size} of $\struct A$ to be the cardinality of its universe.
	
	\subsubsection{Fixed-point Logic with Counting}

Fixed-point logic with counting (\FPC) is an extension of inflationary fixed-point logic with the ability to express the cardinality of definable sets. The logic has two sorts of first-order variable: \emph{element variables}, which range over elements of the structure on which a formula is interpreted in the usual way, and \emph{number variables}, which range over some initial segment of the natural numbers. We  write element variables with lower-case Latin letters $x, y, \dots$ and use lower-case Greek letters $\mu, \eta, \dots$ to denote number variables.

The atomic formulas of $\FPC[\tau]$ are all formulas of the form $\mu
= \eta$ or $\mu \le \eta$, where $\mu, \eta$ are number variables; $s
= t$ where $s,t$ are element variables or constant symbols from
$\tau$; and $R(t_1, \dots, t_m)$, where each $t_i$ is either an
element variable or a constant symbol and $R$ is a relation symbol
(i.e.\ either a symbol from $\tau$ or a relational variable) of arity
$m$.  Each relational variable of arity $m$ has an associated type
from $\{\mathrm{elem},\mathrm{num}\}^m$.  The set $\FPC[\tau]$ of \emph{$\FPC$ formulas} over $\tau$ is built up from the atomic formulas by applying an inflationary fixed-point operator $[\ifpop_{R,\tup x}\phi](\tup t) $; forming \emph{counting terms} $\countingTerm{x} \phi$, where $\phi$ is a formula and $x$ an element variable; forming formulas of the kind $s = t$ and $s \le t$ where $s,t$ are number variables or counting terms; as well as the standard first-order operations of negation, conjunction, disjunction, universal and existential quantification. Collectively, we refer to element variables and constant symbols as \emph{element terms}, and to number variables and counting terms as \emph{number terms}.

For the semantics, number terms take  values in $\{0,\ldots,n\}$,
where $n= \univ{A}$ and
element terms take values in $\dom(\struct A)$. The semantics of atomic formulas,
fixed-points and first-order operations are defined as usual (c.f.,
e.g., \cite{EF99} for details), with comparison of number terms
$\mu \le \eta$ interpreted by comparing the corresponding integers in
$\{0,\ldots,n\}$. Finally, consider a counting term of the form
$\countingTerm{x}\phi$, where $\phi$ is a formula and $x$ an element
variable. Here the intended semantics is that $\countingTerm{x}\phi$
denotes the number (i.e.\ the element of $\{0,\ldots,n\}$) of elements that
satisfy the formula $\phi$.
For a more detailed definition of $\FPC$, we refer the reader
to~\cite{EF99, Lib04}.  

We also consider $\infC$---the infinitary logic with counting, and
finitely many variables.  We will not define it formally (the
interested reader may consult~\cite{Ott97}) but we need the following
two facts about it: its expressive power properly subsumes that of
$\FPC$, and it is closed under $\FPC$-reductions, defined below.

%\marginnote{this paragraph introducing notation replaces what was a lemma.  Is this OK? As far as I could see, the lemma itself is not referenced anywhere.}
It is known by the Immerman-Vardi theorem~\cite{EF99} that fixed-point logic can express all polynomial-time properties of finite ordered structures.  It follows that in $\FPC$ we can express all polynomial-time relations on the number domain.  In particular, we have formulas with free number variables $\alpha,\beta$ for defining sum and product, and we simply write $\alpha+ \beta$ and $\alpha\cdot \beta$ to denote these formulas.   For a number term $\alpha$ and a non-negative integer $m$, we write $\alpha=m$ as short-hand for the formula that says that $\alpha$ is exactly $m$.
We write $\BIT(\alpha,\beta)$ to denote the formula that is true just in case the $\beta$-th bit in the binary expansion of $\alpha$ is $1$.  Finally, for each constant $c$, we assume a formula $\MULT_c(W,x,y)$ which works as follows.  If $B$ is an ordered set and $W\subseteq B$ is a unary relation that codes the binary representation of an integer $w$, then $\MULT_c$ defines a binary relation $R \subseteq B^2$ which on the lexicographic order on $B^2$ defines the binary representation of $c\cdot w$.

	\subsubsection{Reductions}
	
We frequently consider ways of defining one structure within another in some logic $\logic L$, such as first-order logic or $\FPC$.  Consider two signatures $\sigma$ and $\tau$ and a logic $\logic L$. An \emph{$m$-ary $\logic L$-interpretation of $\tau$ in $\sigma$} is a sequence of formulae of $\logic L$ in vocabulary $\sigma$ consisting of:
(i) a formula $\delta(\tup x)$;
(ii) a formula $\varepsilon(\tup x, \tup y)$;
(iii) for each relation symbol $R \in \tau$ of arity $k$, a formula $\phi_R(\tup x_1, \dots, \tup x_k)$; and
(iv) for each constant symbol $c \in \tau$, a formula $\gamma_c(\tup x)$,
where each $\tup x$, $\tup y$ or $\tup x_i$ is an $m$-tuple of free variables. We call $m$ the \emph{width} of the interpretation. We say that an interpretation $\Theta$ associates a $\tau$-structure $\struct B$ to a $\sigma$-structure $\struct A$ if there is a surjective map $h$ from the $m$-tuples $\{ \tup a \in \univ{A}^m \mid \struct A \models \delta[\tup a] \}$ to $\struct B$ such that:

\begin{itemize}
	\item $h(\tup a_1) = h(\tup a_2)$ if, and only if, $\struct A \models \varepsilon[\tup a_1, \tup a_2]$;
	
	\item $R^\struct{B}(h(\tup a_1), \dots, h(\tup a_k))$ if, and only if, $\struct A \models \phi_R[\tup a_1, \dots, \tup a_k]$;
	
	\item $h(\tup a) = c^\struct{B}$ if, and only if, $\struct A \models \gamma_c[\tup a]$.
\end{itemize}

\noindent 
Note that an interpretation $\Theta$ associates a $\tau$-structure with $\struct A$ only if $\varepsilon$ defines an equivalence relation on $\univ{A}^m$ that is a congruence with respect to the relations defined by the formulae $\phi_R$ and $\gamma_c$. In such cases, however, $\struct B$ is uniquely defined up to isomorphism and we write $\Theta(\struct A) \defeq \struct B$. 
Throughout this paper, we will often use interpretations where $\varepsilon$ is simply defined
as the usual equality on $\tup a_1$ and $\tup a_2$. In these instances, we omit the explicit
definition of $\varepsilon$.

The notion of interpretations is used to define logical reductions. Let $C_1$ and $C_2$
be two classes of $\sigma$- and $\tau$-structures respectively. We say $C_1$ \emph{$\logic L$-reduces}
to $C_2$ if there is an $\logic L$-interpretation $\Theta$ of $\tau$ in $\sigma$, such that
$\Theta(\struct{A})\in C_2$ if and only if $\struct{A}\in C_1$, and we write $C_1\leq_{\logic L}C_2$.

It is not difficult to show that formulas of \FPC compose with reductions in the sense that, given an interpretation $\Theta$ of $\tau$ in $\sigma$ and a $\sigma$-formula $\phi$, we can define a
$\tau$-formula $\phi'$ such that $\struct A \models \phi'$ if, and
only if, $\Theta(\struct A) \models \phi$.  Moreover $\infC$ is closed
under $\FPC$-reductions.  So if $C_2$ is definable in $\infC$ and
$C_1 \leq_{\logic L}C_2$, then $C_1$ is also definable in $\infC$.

	% \begin{definition}[Logical Interpretation]
	% 	Given a logic $L$, and two relational vocabularies $\tau$ and $\rho$, 
	% 	with $\rho=(R_1,\ldots,R_k)$, a $L$-\emph{interpretation} of $\rho$ in
	% 	$\tau$ is a sequence $\phi=(\phi_U,\phi_1,\ldots,\phi_k)$ of
	% 	$\tau$-formulas in the logic $L$, where the free variables of $\phi_U$
	% 	are among $\{x_1,\ldots,x_l\}$, and those of $\phi_i$ are among
	% 	$\{x_1,\ldots,x_{ar(R_i)\cdot l}\}$ for some fixed $l$.
		
	% 	For a $\tau$-structure $\struct{A}$ with universe $A$,
	% 	and an interpretation $\phi$, $\phi(A)$ 
	% 	defines a $\rho$-structure $(U,R_1,\ldots,R_k)$ where the universe
	% 	$U\subseteq A^l$ is defined by $\phi_U$, and the relation $R_i$ is defined
	% 	by $\phi_i$. 
		
	% 	Let $C_\tau$ and $C_\rho$ be classes of $\tau$- and $\rho$-structures
	% 	respectively. If there exists a $L$-interpretation $\phi$ of $\rho$
	% 	in $\tau$ such that for any $\tau$-structure $\struct{A}$ it holds
	% 	$\struct{A}\in C_\tau\Leftrightarrow \phi(\struct{A})\in C_\rho$, then
	% 	we say $C_\tau$ $L$-\emph{reduces} to $C_\rho$, or $C_\tau\leq_L C_\rho$.
	% 	We call $\phi$ a $L$-\emph{reduction}.
	% \end{definition}

	\subsubsection{Representation}
	In order to discuss definability of constraint satisfaction and linear programming problems,
	we need to fix a representation of instances of these problems as relational structures. Here, we 
	describe the representation we use.
	
	\noindent \textbf{Numbers and Vectors.} We represent an integer $z$ as a 
	relational structure in the following way. Let $z=s\cdot x$, with $s\in\{-1,1\}$
	being the sign of $z$, and $x\in\mathbb{N}$, and let $b\geq \lceil\log_2(x)\rceil$.
	We represent $z$ as the structure $\struct{z}$ with universe $\{1,\ldots,b\}$
	over the vocabulary $\tau_{\mathbb{Z}}=\{X,S,<\}$, where $<$ is interpreted 
	the usual linear order on $\{1,\ldots,b\}$; $S^{\struct{z}}$ is a unary relation
	where $S^{\struct{z}}=\emptyset$ indicates that $s=1$, and $s=-1$ otherwise;
	and $X^{\struct{z}}$ is a unary relation that encodes the bit representation
	of $x$, i.e.\ $X^{\struct{z}}=\{k\in\{1,\ldots,b\} \mid \BIT(x,k)=1\}$. In a similar
	vein, we represent a rational number $q=s\cdot \frac{x}{d}$ by a structure
	$\struct{q}$ over the domain $\tau_{\mathbb{Q}}=\{X,D,S,<\}$, where the
	additional relation $D^{\struct{q}}$ encodes the binary representation of the
	denominator $d$ in the same way as before.
	
	In order to represent vectors and matrices over integers or rationals, we
	have multi-sorted universes.
	Let $T$ be a non-empty set, and let $v$ be a 
	vector of integers indexed by $T$.  We represent $v$ as a structure $\struct{v}$
	with a two-sorted universe with an index sort  $T$, and bit 
	sorts $\{1,\ldots,b\}$, where $b\geq\lceil\log_2(|m|)\rceil$,
	$m=\max_{t\in T}v_t$, over the vocabulary $(X,D,S,<)$. Now, the
	relation $S$ is of arity $2$, and $S^{\struct{v}}(t,\cdot)$ encodes the sign
	of the integer $v_t$ for $t\in T$.  Similarly, $X$ is a binary relation interpreted
	as $X^{\struct{v}}=\{(t,k)\in T\times \{1,\ldots,b\}\mid \BIT(v_t,k)=1\}$.
	In order to represent 
	matrices $M\in \mathbb{Z}^{T_1\times T_2}$, indexed by two sets $T_1, T_2$,
	we allow three-sorted universes with
	two sorts of index sets. The generalisation to rationals carries over from 
	the numbers case.  We write $\vocvec$ to denote the vocabulary for vectors over $\QQ$ and $\vocmat$ for the vocabulary for matrices over $\QQ$.
	
	\noindent \textbf{Linear Programs.} Let an instance of a linear optimization problem
	be given by a constraint maxtrix $A\in \mathbb{Q}^{C\times V}$,
	and vectors $\tup b\in\mathbb{Q}^C, \tup c\in\mathbb{Q}^V$ over some set of variables
	$V$ and constraints $C$. We represent this instance in the natural way as a structure over
	the vocabulary $\tau_{\LP}=\tau_{vec}\dunion\tau_{mat}$. 

We can now state the result from~\cite{adh13:lics} that we require, to the effect that there is an $\FPC$ interpretation that can define solutions to linear programs.
	\begin{theorem}[Theorem 11, \cite{adh13:lics}]\label{thm:lpdefinable}
		Let an instance $(A\in \mathbb{Q}^{C\times Q}, \tup b\in\mathbb{Q}^C, \tup c\in\mathbb{Q}^V)$
		of a LP be explicitly given by a relational representation in $\tau_{LP}$.
		Then, there is a $\FPC$-interpretation that defines a representation of
		$(f\in\mathbb{Q},\tup v\in\mathbb{Q}^V)$, such that $f=1$ if and only if
		$\max_{\tup x\in P_{A,\tup b}} \tup c^T\tup x$ is unbounded, $\tup v\notin P_{A,\tup b}$ if and
		only if there is no feasible solution, and 
		$f=0, \tup v=\mathrm{argmax}_{\tup x\in P_{A,\tup b}} \tup c^T\tup x$ otherwise.
	\end{theorem}

	\noindent \textbf{CSPs.} 
We next examine how instances of $\VCSP(\Gamma)$ for finite $\Gamma$ are represented as relational structures.  We return to the case of infinite $\Gamma$ in Section~\ref{s:inf}.

	For a fixed finite language
	$\Gamma = \{f_1,\ldots,f_k\}$, we represent an instance $I$ of $\VCSP(\Gamma)$ as a structure
	$\struct{I}=(\univ I,<,(R^{\struct I}_f)_{f\in\Gamma},W_N^\struct{I},W_D^\struct{I})$
	over the vocabulary $\tau_\Gamma$. The universe $\univ I =V\dunion C\dunion B$
	is a three-sorted set, consisting of variables $V$, constraints $C$, and a set $B$ of \emph{bit positions}.  We assume that $|B|$ is at least as large as the number of bits required to represent the numerator and denominator of any rational weight occurring in $I$. The relation $<$ is a linear order on $B$. The 
	relation $R_f^\struct{I} \subseteq V^{ar(f)}\times C$
	contains $(\sigma,c)$ if $c=(\sigma, f,q)$ is a constraint in $I$.  
	The relations $W_N^\struct{I},W_D^\struct{I} \subseteq C\times B$ encode the weights of the
	constraints: $W_N^\struct{I}(c,\beta)$ (or $W_D^\struct{I}(c,\beta)$) holds if and only if the $\beta$-th bit of the 
	bit-representation of the numerator (or denominator, respectively) of the weight of
	constraint $c$ is one.
	For the decision version of the VCSP, we have two additional
	unary relation $T_N$ and $T_D$ in the vocabulary which
	encode the binary representation of the numerator and
	denominator of the threshold constant of the instance.

We are now ready to define what it means to express $\VCSP(\Gamma)$ in a logic such as $\FPC$.  For a fixed finite langauge $\Gamma$, we say that the decision version of $\VCSP(\Gamma)$ is definable in a logic $L$ if there is some $\tau_{\Gamma} \cup \{T_N,T_D\}$-sentence $\phi$ of $L$ such that $\struct I \models \phi$ if, and only if, $I$ is satisfiable.  We say that $\VCSP(\Gamma)$ is definable in $\FPC $ if there is an $\FPC$ interpretaion $\Theta$ of the vocabulary $\tau_{\QQ}$ in $\tau_{\Gamma}$ such that for any $\struct I$, $\Theta(\struct I)$ codes the value of an optimal solution for the instance $I$.

%% file: reductions_full.tex
An essential part of the machinery that leads to Theorem~\ref{thm:dichotomy} is that the computational complexity of $\VCSP(\Gamma)$ is robust under certain changes to $\Gamma$.  In other words, closing the class of functions $\Gamma$ under certain natural operations does not change the complexity of the problem.  This is established by showing that the distinct problems obtained are inter-reducible under polynomial-time reductions.  Our aim in this section is to show that these reductions can be expressed as interpretations in a suitable logic (in some cases first-order logic suffices, and in others we need the power of counting).  

The following lemma is analogous to Lemma~\ref{lem:expower} and shows that the reductions there can be expressed as logical interpretations.
		
	\begin{lemma}\label{lem:expower_to_gamma}
%\marginnote{Do we really prove first-order reductions?  The proof seems to use the BIT predicate.}
		Let $\Gamma$ and $\Gamma'$ be valued constraint languages over domain $D$ of finite
		sizes such that $\Gamma' \subseteq \langle\Gamma\rangle$.
		Then $\VCSP(\Gamma') \le_{\FPC} \VCSP(\Gamma)$.
	\end{lemma}
	\begin{proof}
The construction of the reduction follows closely  the proof of Theorem 3.4.\ in \cite{ccjk06:ai}, while ensuring it is definable in $\FPC$.
		
		Let $I=(V,C)$ be a given instance of $\VCSP(\Gamma')$.
		We fix for 
		each function $f\in\Gamma'$ of arity $m$ an instance $I_f=(V_f,C_f)$ of $\VCSP(\Gamma)$
		and a $m$-tuple of distinct elements $\tup{v}_f\in V_f^m$ that together
		express $f$ in the sense of Definition~\ref{def:expower}. The idea is now
		to replace each constraint
		$c=(\sigma,f,q)\in C$ by a copy of $I_f$ where the variables $v_{f1},\ldots,v_{fm}$
		in $I_f$ are identified with $\sigma_1,\ldots,\sigma_m$, and the 
		remaining variables are fresh. Since each $I_f$ is an instance of $\VCSP(\Gamma)$,
		the instance $J=(U,E)$ obtained after all replacements is again an instance of 
		$\VCSP(\Gamma)$. Furthermore, by Definition~\ref{def:expower} it has the
		same optimal solution as $I$.
		
		Formally, we define the instance $J=(U,E)$ as follows. The set of variables
		$U$ consists of the variables in $V$ plus a fresh copy of the variables in
		$V_f$ for each constraint in $C$ that uses the
		function $f$, so we can identify $U$ with the
		following set.
		\[U=V\dunion \{(v,c)\mid c \in C, v\in V_f\}.\]
		Each constraint $c=(\sigma,f,q)\in C$ gives rise to a set of constraints $E_c$,
		representing a copy of the constraints in $C_f$.
		\[E_c=\{(h_c(\nu),g,q\cdot r)\mid (\nu,g,r)\in C_f\},\]
		where $h_c:V_f\rightarrow U$ is defined as the mapping $h_c(v)=\sigma_i$,
		if $ v = v_{fi}$, and $h_{c}(v)=(v,c)$ otherwise.
		The set of constraints $E$ is then simply the union of all sets $E_c$.
		\[E=\bigcup_{c\in C}E_c.\]
		
		Let $\tau_\Gamma=(<,(R_f)_{f\in\Gamma},W_N,W_D)$ and 
		$\tau_{\Gamma'}=(<,(R_f)_{f\in\Gamma'},W_N,W_D)$ be the vocabularies for
		instances of $\VCSP(\Gamma)$ and $\VCSP(\Gamma')$
		respectively.  We aim to define an $\FPC$ reduction 
		$\Theta=(\tup{\delta},\varepsilon,\phi_{<},(\phi_{R_f})_{f\in\Gamma},\phi_{W_N},\phi_{W_D})$
		such that $\struct J=\Theta(\struct I)$
		corresponds to the above construction of the instance $J$.
		
		Let an instance $I=(V,C)$ of $\VCSP(\Gamma')$ be given
		as a structure $\struct I$ over $\tau_{\Gamma'}$ with the three-sorted
		universe $\univ I=V\dunion C\dunion B$.  For
		each $m$-ary function $f\in\Gamma$ we have fixed an instance 
		$I_f=(V_f,C_f)$ and a tuple
		$\tup{v}_f=(v_{f1},\ldots,v_{fm})$ that together express $f$.  As the construction of $\struct J$ depends on these instances, we fix an encoding of them in an initial segment of the natural numbers.  To be precise, as the sets 
		$\hat V=\bigcup_{f\in\Gamma'}V_f$
		and $\hat C=\bigcup_{f\in\Gamma'}C_f$ are of fixed size (independent of $I$), let $n_{\hat V}=|\hat V|$ and $n_{\hat C}=|\hat C|$.  We then fix bijections $\var:\hat V \ra \{1,\ldots,n_{\hat V}\}$ and $\con:\hat C \ra \{1,\ldots,n_{\hat C}\}$ such that for each $f\in\Gamma'$, there are intervals $\mathcal{V}_f=[lv_f,rv_f]$ and $\mathcal{C}_f =[lc_f,rc_f]$ such that $\var(V_f) = \mathcal{V}_f$ and $\con(C_f) = \mathcal{C}_f$.  We assume that $\univ I$ is larger than $\max(n_{\hat V},n_{\hat C})$ so that we can use number terms to index the elements of $\hat V$ and $\hat C$.  There are only finitely many instances $I$ smaller than this, and they can be handled in the interpretation $\Theta$ individually.

		In defining the formulas below, for an integer interval $I$ we write $\mu \in I$  as shorthand for the formula $\bigvee_{m \in I} \mu = m$.

		The universe of 
		$\struct J$ is a three-sorted set $\univ{J}=U\dunion E\dunion B'$
		consisting
		of variables $U$, constraints $E$, and bit positions $B'$. The set $U$
		is defined by the formula
		\begin{align*}\delta_U(x,\mu)=&
			\left(x\in C \wedge 
			\bigvee_{f\in\Gamma'}(\exists \tup{y}\in V^{ar(f)}: R_f(\tup{y},x)\wedge \mu \in \mathcal{V}_f)\right)
			\\&\vee(\mu = 0 \wedge x\in V).
		\end{align*}
		In other words, the elements of $U$ consist of pairs $(x,\mu)$, where $x \in V \cup C$ and $\mu$ is an element of the number domain and  
		we make the following case distinction: Either
		$x\in C$ and there is a constraint $x=(\tup y, f, q)$ in $I$, and
		a variable $v \in V_f$ with $\var(v) = \mu$; then the pair represents one of the
		fresh variables in $C\times \hat V$. Or, $x\in V$ and $\mu = 0$ and the pair
		simply represents an element of $V$. 
		
		Similarly, the constraints $E$ are given by
		\begin{align*}
			\delta_E(x,\mu)= x\in C \wedge
			\bigvee_{f\in\Gamma'}(\exists \tup{y}\in V^{ar(f)}: R_f(\tup{y},x)\wedge \mu\in \mathcal{C}_f).
		\end{align*}
		Again, the elements of $E$ are pairs $(x,\mu)$, with $\in C$ and $\mu$ an element of the number domain,
		and we require that if there is a constraint of the form $x=(\tup{y},f,q)$,
		then there is a constraint $c \in C_f$ with $\con(c) = \mu$.
		
		For the domain of bit positions, we just need to make sure that the set is large enough to encode all weights in $J$.  Taking $B'=B^2$ suffices, so
		\[\delta_{B'}(x_1,x_2)=x_1,x_2\in B\] 
and we take $\phi_{<}(\tup x, \tup y)$ to be the formula that defines the lexicographic order on pairs.

		The constraints of $J$ are encoded in the relations $R_g$, $g\in\Gamma$. 
		For an $m$-ary function $g$, this is defined by a formula $\phi_{R_g}$ in the free variables $(x_1,\mu_1,\ldots,x_m,\mu_m,e,\nu)$ where each $(x_i,\mu_i)$ ranges over elements of $U$, and $(e,\nu)$ ranges over elements of $E$.  To be precise, we define the formula by:

		\begin{align*}\phi_{R_g} =
			\bigvee_{f\in\Gamma'} &\left(\exists \vec{y} \in V^{ar(f)}:R_f(\tup y,e)
			\wedge \nu \in \mathcal{C}_f \right.
			\\  &\wedge \left.
		 \bigvee_{e' = (\rho,g,r) \in C_f} \left( \nu =\ con(e')  \land 
                 \bigwedge_{i: \rho_i \in \vec{v}_f}  (x_i = e \land  \mu_i = \var(\rho_i)) \right.\right. \\
                 &\land \left.\left.
      \bigwedge_{i: \rho_i \not\in \vec{v}_f} (x_i = y_i \land \mu_i = 0) \right)
\right).
		\end{align*}
		
		Finally, we define the weight relations. The weight of a constraint $\tup e=(e_1,e_2)$
		is assigned the product of the weight of $e_1\in C$ and the weight of $e_2\in \hat C$.
		We have
		\begin{align*}
			\phi_{W_N}(\tup e, \tup \beta)=\bigvee_{e'\in\hat C}e_2=\con(e')
				\wedge \MULT_{w}(W_N(e_1,\cdot),\tup beta),
		\end{align*}
where $w$ is the numerator of the weight of the constraint $e'$.
		The
		definition of the denominator relation is analogous.
	\end{proof}
	
The next lemma similarly establishes that the reduction in Lemma~\ref{lem:equiv} can be realised as an $\FPC$ interpretation.

	\begin{lemma}\label{lem:equiv_to_gamma}
		Let $\Gamma$ and $\Gamma'$ be valued languages over domain $D$ of 
		finite sizes such that
		$\Gamma' \subseteq \Gamma_\equiv$. Then $\VCSP(\Gamma') \le_{\FPC} \VCSP(\Gamma)$.
	\end{lemma}
	\begin{proof}
		Note that adding constants to the value of constraints does not change the
		optimal solution of the instance. Hence, we only need to adapt to the
		scaling of the constraint functions. This can be achieved by changing
		the weights accordingly.
		  
		Let $I=(V,C)$ be an instance of $\VCSP(\Gamma')$, given as the relational structure 
		$\struct{I}=(\univ I,(R_f)_{f\in\Gamma'},W_N,W_D)$.  We aim to construct an instance $J=(U,E)$ of $\VCSP(\Gamma)$
		with the same optimal solution.
		
		The set of variables of $J$ is $V$. For any $f\in\Gamma'$
		we fix a function $S(f) \in \Gamma$ such that $S(f) \equiv f$. Then,
		the formula 
		$\phi_{R_g}(\sigma, d)=\bigvee_{f\in\Gamma';g=S(f)}R_f(\sigma, d)$ 
		defines the constraints of $J$. Let $d=(\sigma,g,r)$ be any constraint
		in $E$, and $c=(\sigma,f,q)$ be the corresponding constraint in $C$ where
		$g=S(f)$, and $g=a\cdot f+b$ for some $a,b\in\QQ$. 
		We then set the weight $r$ of the constraint $d$ to be $a\cdot q$.
		This can again be defined by a formula in $\FPC$.
	\end{proof}

Next, we show that there is a definable reduction from $\VCSP(\Gamma)$ to the problem defined by a core of $\Gamma$.
	\begin{lemma}\label{lem:gamma_to_core}
		Let $\Gamma$ be a valued language over $D$, and $\Gamma'$ a core of 
		$\Gamma$. Then, $\VCSP(\Gamma)\le_{\FO}\VCSP(\Gamma')$.
	\end{lemma}
	\begin{proof}
		Since the functions in $\Gamma'$ are exactly those in $\Gamma$, only 
		restricted to some subset of $D$, we can interpret any instance of 
		$\VCSP(\Gamma)$ directly as an instance of $\VCSP(\Gamma')$. Since the
		optimum of both instances are the same, by Lemma \ref{lem:core_equiv},
		this constitutes a reduction.
	\end{proof}

	The next two
	Lemmas together show that $\VCSP(\Gamma)$ and $\VCSP(\Gamma_c)$ are 
	$\FPC$-equivalent. The proof follows closely the proof from \cite{hkp14:joc} that they are polynomial-time equivalent.
	
	\begin{lemma}[Lemma 2, \cite{hkp14:joc}]\label{lem:perm_vcsp}
		Let $\Gamma$ be a core over domain $D$.
		There exists an
		instance $I_p$ of $\VCSP(\Gamma)$ with variables $V=\{x_a \mid a\in D\}$ such
		that $h_{id}(x_a)=a$ is an optimal solution of
		$I_p$ and for every optimal solution $h$, the following hold:
		\begin{enumerate}
		  \item $h$ is injective; and 
		  \item for every instance $I'$ of $\VCSP(\Gamma)$ and every optimal 
		  		solution $h'$ of $I'$, the mapping $s_h\circ h'$ is also an
		  		optimal solution, where $s_h(a):= h(x_a)$.
		\end{enumerate} 
	\end{lemma}

	\begin{lemma}\label{lem:gammac_to_gamma} Let $\Gamma$ be a core
		over a domain $D$ of finite size. Then,
		$\VCSP(\Gamma_c)\leq_{\FPC}\VCSP(\Gamma)$.
	\end{lemma}
	\begin{proof}		
		Let $I_c=(V_c,C_c)$ be an instance of $\VCSP(\Gamma_c)$, and let
		$I_p=(V_p,C_p)$ be an instance of $\VCSP(\Gamma)$ that satisfies the
		conditions of Lemma~\ref{lem:perm_vcsp}. We construct an instance
		$I=(V,C)$ of $\VCSP(\Gamma)$ as follows.  The set of variables
		$V$ is
		\[V \defeq V_c \dunion V_p = V_c \dunion \{x_a\mid a\in D\}.\]
		By Definition~\ref{def:gamma_c}, each function $f\in \Gamma_c$
		is associated with some function $g=\gamma(f)\in\Gamma$, such that $f$ is obtained
		from $g$ by fixing the values of some set of variables of $g$. Let $T_{f}$
		be the corresponding index set, 
		$t_{f}:T_{f}\rightarrow D$
		the corresponding partial assignment of variables of $g$, and $s_{f}$
		the injective mapping between parameter positions of $f$ and $g$.
		Then, we add for each constraint $c'=(\sigma', f, q)\in C_c$ the constraint
		$c=(\sigma, g, q)$ to $C$, where we replace each parameter of $g$ that is fixed to
		$a\in D$ by the variable $x_a$, or formally, $\sigma_i=x_{t_{f}(i)}$ if $i\in T_{f}$,
		and $\sigma_i = \sigma'_{s_{f}^{-1}(i)}$ otherwise. Additionally, we 
		add each constraint of $C_p$ to $C$ with its weight multiplied
		by some sufficiently large factor $M$ such that every
		optimal solution to $I$, when restricted to $\{x_a\mid a\in D\}$, constitutes
		also an optimal solution to $I_p$. For instance, $M$ can be chosen as 
		$M:=\sum_{(\sigma,g,q)\in C\backslash C_p}q\cdot \max_{f\in\Gamma_c;\tup{x}}f(\tup{x})$.
		Note that since the domain and the constraint language are finite, 
		and the functions are finite valued,
		the value of $\max_{f\in\Gamma_c;\tup{x}}f(\tup{x})$ exists and is a constant. Together,
		the set of constraints $C$ is defined as
		\begin{align*}
			C =& \{(\sigma,g,q)\mid \exists \sigma',f: g=\gamma(f),(\sigma',f,q)\in C_c,\ 
			\forall i\in T_{f}: \sigma_i = t_{f}(i),\ 
			\forall i\notin T_{f}: \sigma_i = \sigma'_{s_{f}^{-1}(i)}\}\\
			&\cup \{(\sigma,g,M\cdot q)\mid (\sigma,g,q)\in C_p\}.
		\end{align*}
		
		In order to see that this construction is a reduction, consider the optimal
		solutions of $I_c$. Each such optimal solution $h_c$ gives rise to an optimal
		solution $h$ of $I$, where $h(x)=h_c(x)$ for $x\in V_c$, and $h(x)=h_{id}(x)$
		for $x\in V_p$. In the other direction, let $h$ be an optimal solution to
		$I$, and its restriction to $V_p$, $h_p:=h_{|V_p}$ is an optimal solution to $I_p$. 
		By Lemma \ref{lem:perm_vcsp}, the operation $s_{h_p}$ is a permutation on $D$,
		and in particular, by repeatedly applying the second part of Lemma \ref{lem:perm_vcsp},
		the inverse permutation $s_{h_p}^{-1}$ is an optimal solution to $I_p$
		as well. Now, again by application of the second part of Lemma \ref{lem:perm_vcsp},
		we can obtain an optimal solution $h':=s_{h_p}^{-1}\circ h$ to $I$, for which
		$h'(x_a)=a$ for each $a\in D$. That means, the restriction of $h'$ to $V_c$
		is an optimal solution to $I_c$.
		
		We now formulate the above construction as an $\FPC$ interpretation.
		
		Let $I_c$ be given as a structure 
		$\struct I_c$ over $\tau_{\Gamma_c}=(<,(R_f)_{\Gamma_c},W_N,W_D)$. Furthermore,
		let $I_p=(V_p,C_p)$ be some fixed instance of $\VCSP(\Gamma)$
		that satisfies the conditions of Lemma~\ref{lem:perm_vcsp}.
		We construct an $\FPC$-interpretation 
		$\Theta=(\vec{\delta},\varepsilon,\phi_{<},(\phi_{R_f})_{f\in\Gamma},\phi_{W_N},\phi_{W_D})$
		that defines $\struct I=\Theta(\struct I_c)$. The universe
		$\univ{I_c}$ is the three-sorted set $V_c\dunion C_c\dunion B_c$.  In the same way, the
		universe of the structure $\struct I$ is  a three sorted
		set $V\dunion C\dunion B$. Just as in the proof of 
		Lemma~\ref{lem:expower_to_gamma}, to code elements of $V_p$ and $C_p$,
		we fix bijections $\var:V_p \ra \{1,\ldots,|V_p|\}$ and $\con:C_p \ra \{1,\ldots,|C_p|\}$
		
		The set $V$ is then defined by the formula
		\[\delta_V(x)=x\in V_c \lor  x \in \{1,\ldots,|V_p|\}.\]
		
		Similarly, we define $C$ by
		\[\delta_C(x)=x\in C_c \lor  x \in \{1,\ldots,|C_p|\}.\]
		
		The set of bit positions is chosen to be large enough to encode all weights.
		We can choose $B=B_c^2$.
		\[\delta_B(x_1,x_2)=x_1,x_2\in B_c,\]
and let $\phi_{<}$ define the lexicographic order on $B_C^2$.

		For each $m$-ary function $g\in\Gamma$, we have the formula
		\begin{align*}
			\phi_{R_g}(\vec{x},c)= &\bigvee_{e=(\rho,g,r)\in C_p} 
				\left(c=con(e) \land \bigwedge_{1\leq i\leq m}x_i=var(\rho_i) \right)\\
				&\bigvee_{f: \gamma(f) = g}\left(
				\exists \vec{y} \in V_c^{\ar(f)}: R_f(\bar{y},c) \wedge
				\bigwedge_{i\in T_f}x_i=var(t_f(i))
				\bigwedge_{i\notin T_f} x_i=y_{s^{-1}_f(i)}\right).
		\end{align*}
		
		The weights are given by
		\[\phi_{W_N}(c,\tup{\beta}) = (c\in C_c \wedge W_N(c,\beta) ) \lor 
			 \bigvee_{e=(\rho,g,r)\in C_p} (c=con(e) \wedge \MULT_{r\cdot L}(B_c,\tup \beta)),\]
		where $L$ is given by
		\[L = \max_{f\in\Gamma_c;\tup{x}\in D^{ar(f)}}f(\tup x).\]

		The denominator is given by
		\[\phi_{W_D}(c,\tup \beta)=(c\in C_c \wedge W_D(c,\beta)) \lor \bigvee_{e\in C_p} 
			(c=con(e) \wedge \BIT(1,\beta)).\]
		Here, another case distinction is in place. Either we have $c\in C_c$,
		and the weight is simply the same as given by $W_N$ and $W_D$. Or, the
		constraint $c$ corresponds to some constraint $e=(\rho,g,r)\in C_p$, and
		we assign the weight $L\cdot 2^{|B_c|}\cdot r$ to $c$.
	\end{proof}

%% file: expressibility_full.tex
	The fact that $\VCSP(\Gamma)$ is definable in $\FPC$ whenever $\Gamma_c$ does not have the $(\XOR)$ property is obtained quite directly from Theorems~\ref{thm:dichotomy} and~\ref{thm:lpdefinable}.  Here we state the result in somewhat more general form. 

	\begin{theorem}\label{thm:exp}
For any valued constraint language  $\Gamma$ over a finite domain $D$, there is an $\FPC$ interpretation $\Theta$ of $\tau_{\QQ}$ in $\tau_{\Gamma}$ that takes an instance $I$ to a representation of the optimal value of $\BLP(I)$.
	\end{theorem}
\begin{proof}
	We show that it is possible to interpret $\BLP(I)$ as a $\tau_{LP}$-structure in $I$ by means of an $\FPC$-interpretation.  The statement then follows by Theorem \ref{thm:lpdefinable} and the composition of $\FPC$-reductions.
	
	Let $I=(V,C)$ be given as the $\tau_\Gamma$ structure 
	$\struct{I}$ with universe $\univ I=V\dunion C\dunion B$. Our goal is to define a 
	$\tau_{LP}$-structure $\struct{P}$ representing $\BLP(I)$ in
	given by $(A,\tup b, \tup c)$.
	The set of variables of $\struct{P}$ is the union of the two sets 
	\[\lambda=\{\lambda_{c,\nu}\mid c=(\sigma,f,q)\in C, \nu\in D^{|\sigma|}\}\]
	and
	\[\mu=\{\mu_{x,a}\mid x\in V, a\in D\}.\] 
In order to refer to elements 
		of $D$ in our interpretation, we fix a bijection $\dm: D \ra \{1,\ldots,|D|\}$ between $D$ and an initial segment of the natural numbers. 
	
	Then, the sets $\lambda$ and $\mu$ are defined by
	\[\lambda(c,\vec{\nu})=\bigvee_{f\in\Gamma}\left(\exists \vec{y} \in V^{\ar(f)}:R_f(\vec{y},c) \land 
		\bigwedge_{1\leq i\leq ar(f)} \bigvee_{a\in D}\nu_i=\dm(a)\right).\]
Here, we assume that $\vec{\nu}$ is a tuple of number variables of length $\max_{f \in \Gamma} \ar (f)$.  This creates some redundant variables, related to constraints whose arity is less than the maximum.  We also have
	\[\mu(x,\alpha)=x\in V \wedge  \bigvee_{a\in D}y=\dm(a).\]

	For the set of linear constraints, we observe that the constraints resulting
	from the equalities of the form $(2)$ can be indexed by the set
	\[J_{(2)}=\{j_{c,i,a,b}\mid c=(\sigma,f,q)\in C, i\in \{1,\ldots,|\sigma|\},a\in D, b\in \{0,1\}\},\]
	since we have for each $c\in C$, $i\in\{1,\ldots,|\sigma|\}$, and $a\in D$
	a single equality, and hence two inequalities, one for each value of $b$.
	This can be expressed by
	\begin{align*}
		J_{(2)}(c,\iota,\alpha,\beta) =& c\in C \wedge \bigvee_{f\in\Gamma}\exists \vec{y} \in V^{ar(f)}: R_f(\vec{y},c)\\ 
		&\wedge\iota \leq ar(f)\\
		&\wedge  \bigvee_{a\in D}\alpha=\dm(a)\\
		&\wedge \beta \in \{0,1\}.
	\end{align*}
	
	Similarly, the constraints resulting from $(3)$ can be indexed by 
	\[J_{(3)}=\{j_{x,b}\mid x\in V, b\in\{0,1\}\}.\] Or, as a formula,
	\[J_{(3)}(x,\beta) = x\in V \wedge \beta \in\{0,1\}.\]
	Finally, we have two inequalities bounding the range of each variable,
	indexed by
	\[J_{(4)}=\{j_{v,b}\mid v\in \lambda\cup\mu, b\in\{0,1\}\},\] defined by
	\[J_{(4)}(\tup v,\beta)=\lambda(\tup{v})\vee \mu(\tup{v}) \wedge \beta\in\{0,1\}.\] 
	
	The universe $\univ{L}$ is then the three-sorted set
	$Q\dunion R\dunion B'$ with index sets $Q$ and $R$ for columns and rows
	respectively, and a domain for bit positions $B'$, defined by
	\[\delta_Q(\tup x)=\lambda(\tup x)\vee \mu(\tup x),\]
	\[\delta_R(\tup x)=J_{(2)}(\tup x)\vee J_{(3)}(\tup x)\vee J_{(4)}(\tup x),\]
	\[\delta_{B'}(x)=x\in B.\]
	
	The entries in the matrix $A\in\mathbb{Q}^{Q\times R}$, and the 
	two vectors $\tup b\in\mathbb{Q}^{Q}$ and $\tup c\in\mathbb{Q}^{R}$ consist
	only of elements of $\{0,1,-1\}$ and the weight of some constraint in $C$. 
	It is easily seen that these can be suitably defined in $\FPC$.
\end{proof}

Combining this with Theorem~\ref{thm:dichotomy} yields immediately the positive half of the definability dichotomy.
\begin{corollary}\label{cor:exp}
If $\Gamma$ is a valued constraint language such that property $(\XOR)$ does not hold for $\Gamma_c$, then $\VCSP(\Gamma)$ is definable in $\FPC$.
\end{corollary}

%% file: inexpressibility_full.tex
	We now turn to the other direction and show that if
	$\VCSP(\Gamma)$ is such that $\Gamma_c$ has the $(\XOR)$ property then $\VCSP(\Gamma)$ is not definable in $\FPC$.  In fact, we will prove the stronger inexpressibility result
	that those $\VCSP$s are not even definable in the stronger logic
	$\infC$.
	
	Our proof proceeds as follows. The main result in \cite{tz13:stoc} characterizes the
	intractable constraint languages $\Gamma$ as exactly those languages
	whose extension $\Gamma_c$ has the property $(\XOR)$,
	by constructing a polynomial time reduction from $\MAXCUT$ to
	$\VCSP(\Gamma)$.  We show that this reduction can also be carried out 
	within $\FPC$. It is then left to show that $\MAXCUT$ itself is not
	definable in $\infC$. To this end, we describe a series of $\FPC$-reductions
	from 3-$\SAT$ to $\MAXCUT$ which roughly follow their classical
	polynomial time counterparts.  Finally, results of \cite{bjk05:joc}
	and \cite{abd09:tcs} establish that 3-$\SAT$ is not definable
	in $\infC$, concluding the proof.
	
	We consider the problem $\MAXCUT$, where one is given an undirected graph
	$G=(V,E)$ along with a weight function $w:E\rightarrow \QQ^+$ and
	is looking for a bipartition of vertices $p:V\rightarrow \{0,1\}$ that
	maximises the payout function $b(p)=\sum_{(u,v)\in E; p(u)\neq p(v)}w(u,v)$.
	In the decision version of the problem, an additional constant $t\in\QQ^+$
	is given and the question is then whether there is a partition $p$ with
	$b(p)\geq t$.
	
		An instance of (decision) $\MAXCUT$ is given as a relational structure
		$\struct I$ over the vocabulary
		$\tau_\MAXCUT=(E,<,W_N,W_D,T_N,T_D)$. The universe $\univ I$
		is a two-sorted set $U=V\dunion B$,
		consisting of vertices $V$, and a set $B$ of bit positions, linearly ordered by $<$.  In
		addition to the edge relation $E\subseteq V\times V$, there are two
		weight relations $W_N,W_D\subseteq V\times V\times B$ which encode the
		numerator and denominator of the weight between two vertices. Finally,
		the unary relations $T_N,T_D\subseteq B$ encode the numerator and denominator of the
		threshold constant of the instance. 
		
	\begin{lemma}\label{lem:maxcut_to_gamma}
		Let $\Gamma$ be a language over $D$ for which $(\XOR)$ holds. Then,
		$\MAXCUT\leq_{\FPC}\VCSP(\expGamma_{\equiv})$.
	\end{lemma}
	\begin{proof}

		Let $I=(V,E,w,t)$ be a given $\MAXCUT$ instance. 
		We define an equivalent instance 
		$J=(U,C,t')$
		of $\VCSP(\Gamma_{\equiv})$ as follows. 
		Since $(\XOR)$ holds for $\Gamma$, there are two distinct elements 
		$a,b\in D$ for which 
		$\expGamma_{\equiv}$ contains a binary function $f$, such that $f(a,b)=1$ if
		$a=b$ and $f(a,b)=0$ otherwise.  By creating a variable for each
		vertex in $V$ and adding a constraint $((u,v),f,w(e))$ 
		for each edge $e=(u,v)\in E$, 
		we obtain a $\VCSP$ with the same optimal solution. 
		The threshold constant $t'$ is then set to $t'=M-t$, 
		where $M:=\sum_{e\in E}w(e)$.

		We now define a $\FPC$-interpretation $\Theta$ of $\tau_{\expGamma_\equiv}$ 
		in $\tau_{\MAXCUT}$ that carries out the construction. Let $\struct{I}$
		be the relational representation of $I$ over $\tau_{\MAXCUT}$
		with the two-sorted universe $V\dunion B$.
		
		The structure $\struct J = \Theta (\struct I)$ has a three-sorted universe $\univ{J}=U\dunion C\dunion B'$ consisting 
		of variables $U=V$, constraints $C=V^2$, and bit
		positions $B'=B \times \{1,\ldots,|E|\}$.
		\[\delta_U(x) = x\in V,\]
		\[\delta_C(x_1,x_2)=x_1,x_2\in V,\]
		\[\delta_{B'}(x,\mu)=x\in B \land \mu \leq \#_{y,z} E(y,z).\]
Since $M \leq |E|\max_{e\in E} w(e)$, and each $w(e)$ can be represented by $|B|$ bits, $|E|\cdot |B|$ bits suffice to represent the threshold $M-t$.

		Each edge $e=(u,v)$ gives rise to a constraint $((u,v),e,w(e))$, which is
		then encoded in $R_f$.
		\[\phi_{R_f}(\tup x,\tup c) = E(\tup x) \wedge \bar{x}=\bar{c}.\]
		The weights are simply carried over.
		\[\phi_{W_N}(\bar{c},b) = W_N(\tup c,b)\]
		\[\phi_{W_D}(\bar{c},b) = W_D(\tup c,b)\]
		The threshold is set to $M-t$.  As $\FPC$ can define any polynomial-time computable function on an ordered domain, it is possible to write formulas $\phi_{T_N}$ and $\phi_{T_D}$ defining the numerator and denominator of the threshold $M-t$ on the ordered sort $B'$.

		% \[\phi_{T_N}(b)=\BIT\left(\NUM\left(\sum_{(u,v)\in E}w(u,v)-t\right),b\right)\]
		% \[\phi_{T_D}(b)=\BIT\left(\DENOM\left(\sum_{(u,v)\in E}w(u,v)-t\right),b\right),\]
		% where $w(u,v)$ can be expressed by 
		% $w(u,v)=\frac{\sum_{b=1}^{|B|} 2^b\cdot W_N(u,v,b)}{\sum_{b=1}^{|B|} 2^b\cdot W_D(u,v,b)}$
		% (and analogously for $t$).
                The remaining relations $R_g$ corresponding to
		functions in $g\in\expGamma_\equiv\backslash\{f\}$ are simply empty.
	\end{proof}

	The next ingredient is to show that the classical series
	of polynomial time reductions from 
	$3$-$\SAT$ to $\MAXCUT$ can also be carried out within $\FPC$.
	The chain of reductions goes over three steps, the first one
	reduces $3$-$\SAT$ to $4$-$\NAESAT$ (Not All Equal SAT), then $4$-$\NAESAT$
	is reduced to $3$-$\NAESAT$, and finally $3$-$\NAESAT$ is reduced to
	$\MAXCUT$. We begin with defining the relational representations of
	these problems.
	
	An instance of $3$-$\SAT$ is given as a relational structure 
	over the vocabulary $\tau_{3\SAT}=(R_{000},\ldots,R_{111})$
	with eight ternary relations, one for each possible set
	of negations of literals within a clause (e.g.\ $(a,b,c)\in R_{000}$ 
	may represent the clause $(a\vee b\vee c)$ while $(a,b,c)\in R_{101}$
	may represent $(\neg a\vee b\vee \neg c)$). Similarly, we assume
	$3$-$\NAESAT$ instances to be given as structures over 
	$\tau_{3\NAESAT}=(N_{000},\ldots,N_{111})$, where
	$(a,b,c)\in N_{000}$ represents the constraint that not all of $a,b$ and $c$
	must evaluate to the same value. Finally, a $4$-$\NAESAT$ instance
	is represented as a structure over $\tau_{4\NAESAT}=(N_{0000},\ldots,N_{1111})$,
	only now with sixteen 4-ary relations encoding the clauses.
	
	\begin{lemma}\label{lem:3sat_to_maxcut}
		$3$-$\SAT\leq_{\FPC}\MAXCUT$.
	\end{lemma}
	\begin{proof}
		$3$-$\SAT\leq_{\FO}4$-$\NAESAT$: Let $\struct I=(V,R_{000},\ldots,R_{111})$ be
	any given $3$-$\SAT$ instance. Consider a $4$-$\NAESAT$ instance
	$\struct{J}=(U,N_{0000},\ldots,N_{1111})$ with $V\subset U$, i.e.\ there is
	at least one variable in $U$ not contained in $V$. Furthermore, let 
	$(a,b,c,z)\in N_{ijk0}$ hold if, and only if, $(a,b,c)\in R_{ijk}$ and 
	$z\in U\backslash V$, and let the relations $N_{ijk1}$ be empty. The instance
	$\struct{J}$ is now satisfiable if, and only if, $\struct{I}$ is satisfiable: Whenever there
	is a satisfying assignment for $\struct{I}$, the same assignment 
	extended with $z=0$ for all $z\in U\backslash V$
	will also be a satisfying assignment for $\struct{J}$. In the other direction, 
	if there is a satisfying assignment for $\struct{J}$, there is always a 
	satisfying one that sets $z=0$ for all $z\in U\backslash V$, since 
	negating every variable does not change the value of a $\NAE$-clause, and
	each clause only contains one variable in $U\backslash V$. In terms of
	a $\FPC$-interpretation, this construction looks as follows.
	
	We take as universe $\univ{J}$ the set $V^2$, and interpret an element $(a,a)$ as
	representing the variable $a\in V$, and any element $(a,b), a\neq b$ as
	a fresh variable in $U\backslash V$.
	\[\delta_U(x_1,x_2) = x_1,x_2\in V \]
	\[\phi_{N_{ijk0}}(\bar{x},\bar{y},\bar{z},\bar{w}) = R_{ijk}(x_1,y_1,z_1)
		\wedge w_1 \neq w_2 \land
		\bigwedge_{\bar{v}\in\{\bar{x},\bar{y},\bar{z}\}} v_1=v_2\]
	\[\phi_{N_{ijk1}}(\bar{x},\bar{y},\bar{z},\bar{w}) = \mathrm{False}\]

	\smallskip
	$4$-$\NAESAT\leq_{\FPC}3$-$\NAESAT$: Let $\struct{I}=(V,N_{0000},\ldots,N_{1111})$
	be an instance of $4$-$\NAESAT$. Note that we can split every clause
	$\NAE(a,b,c,d)$ into two smaller $3$-$\NAESAT$ clauses $\NAE(a,b,z)$ and
	$\NAE(\neg z,c,d)$ for some fresh variable $z$. The following interpretation
	realises this conversion. 
	
	In order to introduce a fresh variable for each 
	clause of the $4$-$\NAESAT$ instance, the universe of the $3$-$\NAESAT$
	instance will consist of tuples from $V^4\times \{0,1\}^5$, where
	the first eight components encode a clause in $\struct I$ and the
	last component is a flag indicating whether the element represents a fresh
	variable or one that appears already in $V$. The convention is then that
	an element of the form $(a,a,a,a,0,\ldots,0)$ represents the variable
	$a\in V$, and an element of the form $(a,b,c,d,i,j,m,n,1)$ represents
	the fresh variable that is used to split the clause $N_{ijmn}(a,b,c,d)$.
	\[\delta(\bar{x}) = \bar{x}\in V^4 \times \{0,1\}^5 \]
	\begin{align*}
		\phi_{N_{ij1}}(\bar{x},\bar{y},\bar{z}) =
		 &\bigvee_{m,n\in\{0,1\}}\exists u,v\in V: N_{ijmn}(x_1,y_1,u,v) \\
		 &\wedge\ x_1=x_2=x_3=x_4 \wedge x_5=\ldots =x_9=0\\
		 &\wedge\ y_1=y_2=y_3=y_4 \wedge y_5=\ldots =y_9=0 \\
		 &\wedge\ \bar{z} = (x_1,y_1,u,v,i,j,m,n,1)
	\end{align*}
	\begin{align*}
		\phi_{N_{0ij}}(\bar{x},\bar{y},\bar{z}) =
		 &\bigvee_{m,n\in\{0,1\}}\exists u,v\in V: N_{mnij}(u,v,y_1,z_1) \\
		 &\wedge\ y_1=y_2=y_3=y_4 \wedge y_5=\ldots =y_9=0 \\
		 &\wedge\ z_1=z_2=z_3=z_4 \wedge z_5=\ldots =z_9=0 \\
		 &\wedge\ \bar{x} = (u,v,y_1,z_1,m,n,i,j,1)
	\end{align*}
	The remaining relations are defined as empty.
	
	\smallskip
	$3$-$\NAESAT\leq_{\FPC}\MAXCUT$: The following construction transforms a
	given $3$-$\NAESAT$ instance $\struct{I}=(V,N_{000},\ldots,N_{111})$ into an
	equivalent (decision) $\MAXCUT$ instance $\struct{J}=(\univ J,E,W_N,W_D,T_N,T_D)$. Let
	$m$ be the number of clauses in $\struct I$, and fix $M:=10m$.
	For each variable $v\in V$, we have two vertices denoted $v_0$ and $v_1$, in our graph, along with an edge $(v_0,v_1)$ of weight $M$.
	For each tuple $(x,y,z)\in N_{ijk}$ we add a triangle
	between the vertices $x_i$, $y_j$, and $z_k$ with 
	edge-weight $1$. Setting the cut threshold to $t:=|V|\cdot M+2m$ gives us 
	an equivalent instance: If $\struct{I}$ is satisfiable, say by an assignment $f$,
	then the partition given by $p(v_i)=f(v)+i\ \mathrm{mod}\ 2$ cuts through every edge of the form $(v_0,v_1)$, and through
	two edges in every triangle, resulting in a payout of $|V|\cdot M+2m$. On the
	other hand, any bipartition of payout larger or equal to $|V|\cdot M+2m$
	has to cut through all edges of the form $(v_0,v_1)$, since it can only
	cut through two edges in each triangle. 
	Hence, any such bipartition induces a satisfying assignment to the
	$3$-$\NAESAT$ instance. We use the following 
	$\FPC$-interpretation to realise this construction.
	
	The universe of $\struct{J}$ is defined as a two-sorted set $\univ J=U\dunion B$,
	consisting of vertices $U=V\times\{0,1\}$ and bit positions $B=\{1,\ldots,\alpha\}$
	for some sufficiently large $\alpha$. In particular, $\alpha$ has to be chosen larger
	than $\log_2 t$. Since $m$ is at most $|V|^3$, taking $\alpha=|V|^4$ suffices.
	\[\delta_U(x_1,x_2)=x_1\in V, x_2\in \{0,1\} \]
	\[\delta_B(\vec{\mu})= \bigwedge_{1\leq i \leq 4} \mu_i \leq \#_v v \in V.\]
	
	The edge relation is given by
	\begin{align*}
		\phi_E(\bar{x},\bar{y}) &= x_1=y_1 \wedge x_2\neq y_2\\
			& \bigvee_{i,j,k\in\{0,1\}}\exists u,v,w\in V: N_{ijk}(u,v,w)
			\wedge \bar{x},\bar{y}\in \{(u,i),(v,j),(w,k)\} .
	\end{align*}
	
	The edge weights and the cut threshold are defined by
	\begin{align*}
		\phi_{W_N}(\bar{x},\bar{y},\beta) &= x_1=y_1 \wedge x_2\neq y_2 \wedge \BIT(1,\beta) \\
			&\vee \BIT\left(10 \cdot \sum_{i,j,k\in\{0,1\}}\#_{u,v,w}N_{ijk}(u,v,w),\beta\right),
	\end{align*}
	\[\phi_{T_N}(\beta) = \BIT\left((2+10\cdot \#_v v\in V) \cdot \sum_{i,j,k\in\{0,1\}}\#_{u,v,w}N_{ijk}(u,v,w),\beta\right),\]
	\begin{align*}
		\phi_{W_D}(\bar{x},\bar{y},\beta) = \BIT(1,\beta),
	\end{align*}
	\[\phi_{T_D}(\beta) = \BIT(1,\beta).\]
	Note that the weights and the cut threshold are integer, 
	hence the denominator relation are simply coding $1$.
	
	\end{proof}
	
	\begin{lemma}\label{lem:3sat}
		$3$-$\SAT$ is not expressible in $\infC$.
	\end{lemma}
	\begin{proof}
		Note that a $3$-$\SAT$ instance $\struct{I}=(V,R^{\struct{I}}_{000},\ldots,R^{\struct{I}}_{111})$
		can also be interpreted as an instance of $\CSP(\Gamma_{3\SAT})$ for 
		$\Gamma_{3SAT}=\{R_{000},\ldots,R_{111}\}$ and $R_{ijk}=\{0,1\}^3\backslash (i,j,k)$.
		Hence, we can apply results from the algebraic classification of CSPs to 
		determine the definability of $3$-$\SAT$. In this context, it has been
		shown in \cite{bjk05:joc} that the algebra of polymorphisms
		corresponding to $\Gamma_{3SAT}$ contains
		only essentially unary operations. It follows 
		from the result in \cite{abd09:tcs} that $3$-$\SAT$ is not definable in $\infC$.
	\end{proof}
	
	\begin{theorem}\label{thm:inexp}
		Let $\Gamma$ be a valued constraint language of finite size
		and let $\Gamma'$ be a core 
		of $\Gamma$. If $(\XOR)$ holds for $\Gamma'_c$, then $\VCSP(\Gamma)$
		is not expressible in $\infC$.
	\end{theorem}
	\begin{proof}
		Assume property $(\XOR)$ holds for $\Gamma'_c$. By Lemma 
		\ref{lem:maxcut_to_gamma}, $\MAXCUT$ $\FPC$-reduces to $\VCSP(\langle\Gamma'_c\rangle_{\equiv})$.
		Lemmas \ref{lem:expower_to_gamma} to \ref{lem:gammac_to_gamma} provide
		a chain of $\FPC$-reductions from $\VCSP(\langle\Gamma'_c\rangle_{\equiv})$
		to $\VCSP(\Gamma)$. Since $\infC$ is closed under
		$\FPC$-reductions, Lemmas \ref{lem:3sat_to_maxcut} and \ref{lem:3sat}
		together show that $\MAXCUT$ is not definable in $\infC$, and hence neither
		is $\VCSP(\Gamma)$.
	\end{proof}

%% file: inf_size_full.tex
 In representing the problem $\VCSP(\Gamma)$ as a class of relational structures, we have chosen to fix a finite relational signature $\tau_{\Gamma}$ for each finite $\Gamma$.  An alternative, \emph{uniform} representaation would be to fix a single signature which allows for the representation of instances of $\VCSP(\Gamma)$ for arbitrary $\Gamma$ by coding the functions in $\Gamma$ explicitly in the instance.  In this section, , we give a description of how this can be done.  Our goal
	is to show that our results generalise to this case, and that
	the definability dichotomy still holds.

	Let $\Gamma$ now be a valued constraint language over some finite domain $D$.
	The challenge of fixing a relational signature for instances of 
	$\VCSP(\Gamma)$ is that different instances may use different sets of 
	functions of $\Gamma$ in their constraints, and hence, we cannot represent
	each function as a relation in the signature.  Instead, we make the 
	functions part of the universe, together with tuples over $D$ of different 
	arities as their input.  Let $I$ be an instance of $\VCSP(\Gamma)$ where the
	constraints use functions from a finite subset $\Gamma_I\subset\Gamma$, and
	let $m$ be the maximal arity of any function in $\Gamma_I$.
	We then represent $I$ as a structure $\struct{I}$ with the multi-sorted
	universe $\univ I=V\dunion C\dunion B\dunion F\dunion T$, where $V$
	is a set of variables, $C$ a set of constraints, $B$ a set of numbers on 
	which we have a linear order, $F$ a set of function symbols corresponding to
	functions in $\Gamma_I$, and $T$ is a set of tuples from 
	$D\cup D^2\cup\ldots\cup D^m$, over the signature
	$\tau_D=(<,R_{\mathrm{fun}},R_{\mathrm{scope}},W_N,W_D,\Def_N,\Def_D,\Enc)$. Here, the relations
	encode the following information.
	\begin{itemize}
	  \item $R_{\mathrm{fun}}\subseteq C\times F$: This relation matches functions and constraints,
		i.e.\ $(c,f)\in R_{\mathrm{fun}}$ denotes that $c=(\sigma,f,q)$ is a constraint
		of the instance for some scope $\sigma$ and weight $q$.
	  \item $R_{\mathrm{scope}}\subseteq C\times V\times B$: This relation fixes the scope
	  		of a constraint, i.e.\ $(c,v,\beta)\in R_{\mathrm{scope}}$ denotes that 
	  		$c=(\sigma,f,q)$ is a constraint for some function $f$ and weight $q$, where the $\beta$-th
			component of $\sigma$ is $v$.
      \item $W_N,W_D\subseteq C\times B$: This is analogous to the finite case. These
      		two relations together encode the rational weights of the constraints.
      \item $\Def_N,\Def_D\subseteq F\times T\times B$: These two relations together
      		fix the definition of some function symbol in $F$. That is,
      		$(f,t,\beta)\in \Def_D$ denotes that the $\beta$-th bit of the numerator of
      		the value of $f$ on input $t$ is $1$, and similarly for $\Def_D$ and
      		the denominator.
      \item $\Enc\subseteq T\times D\times B$: This relation fixes the encoding 
      		of tuples as elements in $T$, i.e.\ $(t,a,\beta)\in \Enc$ denotes that
      		the $\beta$-th component of the tuple $t$ is the element $a\in D$.
	\end{itemize}
	The above signature allows now for instances $I$, $I'$ with different sets of
	functions $\Gamma_I$ and $\Gamma_{I'}$ to be represented as structures of 
	the same vocabulary. Since the set of function symbols is part of the 
	universe, the relations $\Def_N, \Def_D$ are required to give concrete meaning
	to these function symbols.
	
We now say, for a (potentially infinite) valued constraint language $\Gamma$ that $\VCSP(\Gamma)$ is \emph{uniformly definable} in $\FPC$ if there is an $\FPC$-interpretation of $\tau_{\QQ}$ in $\tau_D$ which takes an instance $\struct I$ of $\VCSP(\Gamma)$ to the cost of its optimal solution.  Our inexpressibility result, Theorem~\ref{thm:inexp},  immediately carries over to this setting as it is easy to construct an $\FPC$ reduction from the $\tau_{\Gamma}$ representation of $\VCSP(\Gamma)$ to the $\tau_D$ representation.  
\begin{theorem}\label{thm:inf-inexp}
		Let $\Gamma$ be a valued constraint language and let $\Gamma'$ be a core 
		of $\Gamma$. If $(\XOR)$ holds for $\Gamma'_c$, then $\VCSP(\Gamma)$
		is not uniformly definable in $\infC$.
	\end{theorem}

For the positive direction, i.e.\ to show that $\VCSP(\Gamma)$ is \emph{uniformly definable}
in $\FPC$ in all other cases, we simply need to adapt the proof of Theorem \ref{thm:exp}
to fit the new representation.
\begin{theorem}\label{thm:inf-exp}
  		Let $\Gamma$ be a valued constraint language and let $\Gamma'$ be a core 
		of $\Gamma$. If $(\XOR)$ does not hold for $\Gamma'_c$, then $\VCSP(\Gamma)$
		is uniformly definable in $\infC$.
\end{theorem}
\begin{proof}
	We adapt the proof of Theorem \ref{thm:exp} for potentially infinite
	 languages $\Gamma$. The main challenge is to work around the variable
	arities of the constraints.
	
	Let $\Gamma$ be a constraint language over
	some finite domain $D$, and let $I$ be an instance of $\VCSP(\Gamma)$,
	and $\struct{I}$ its relational representation in $\tau_D$.
	Recall that the set of variables of $\BLP(I)$ for 
	$I=(V,C)$ is given by the union of the two sets
	\[\lambda=\{\lambda_{c,\nu}|c=(\sigma,f)\in C, \nu\in D^{|\sigma|}\}\]
	and
	\[\mu=\{\mu_{x,a}|x\in V, a\in D\}.\]
	These sets can now be $\FPC$ defined from $\struct{I}$ as follows.
	\[\lambda(c, s)=\exists f\in F:R_{\mathrm{fun}}(c,f) \wedge \exists \beta\in B:
		\mathrm{Ar}_{R}(c,\beta)=\mathrm{Ar}_{\Enc}(s,\beta)\]
	and
	\[\mu(x,a)=x\in V \wedge a\in T \wedge \mathrm{Ar}_{\Enc}(a,1).\]
	Here, we make use of a formula $\mathrm{Ar}(x,\beta)$ by which we mean that the tuple
	encoded by the element $x$ has the arity $\beta$. The formula can be defined 
	as follows.
	\[\mathrm{Ar}_{R}(x,\beta)=\exists v\in V: \left(R_{\mathrm{scope}}(x,v,\beta)
		\wedge \forall u\in V, \beta'\in B: \beta'\geq \beta\Rightarrow \neg R_{\mathrm{scope}}(x,u,\beta')\right),\]
	\[\mathrm{Ar}_{\Enc}(x,\beta)=\bigvee_{a\in D} \left(\Enc(x,a,\beta)
		\bigwedge_{a'\in D} \forall \beta'\in B: \beta'\geq \beta\Rightarrow \neg \Enc(x,a',\beta')\right).\]
	In words, the formulas ensure that $x$ is a tuple element that is used in the
	structure, that its $\beta$-th component is non-empty, and that for any position 
	$\beta'\geq \beta$, the $\beta'$-th component of $x$ is not defined in the structure.
	
	The rest of the proof follows closely to the original one in Theorem \ref{thm:exp}
	without substantial changes. 
\end{proof}